\documentclass[onecolumn,journal,twoside]{IEEEtran}
\usepackage[utf8]{inputenc}
\usepackage{amsmath}  
\usepackage{tabularx} 
\usepackage{graphicx} 
\usepackage{tikz}
\usepackage[margin=1in,letterpaper]{geometry}  

\usepackage{hyperref}
\usepackage[symbols]{}
\usepackage[epstopdf]{}
\usepackage{amsmath}
\usepackage{amsthm} 
\usepackage{amssymb}
\newcommand{\Prob}{\mathbb{P}}
\newcommand{\var}{\text{Var}}
\newcommand{\Expc}{\mathbb{E}}

\usepackage[T1]{fontenc}

\makeatletter
\def\thm@space@setup{\thm@preskip=2pt
        \thm@postskip=2pt \itshape}
\makeatother

\begin{document}
\setlength{\belowcaptionskip}{-6pt}
        \setlength{\abovedisplayskip}{1mm}
        \setlength{\belowdisplayskip}{1mm}
        \setlength{\abovecaptionskip}{1mm}

\newtheoremstyle{newstyle}
{} 
{} 
{\mdseries} 
{} 
{\bfseries} 
{.} 
{ } 
{} 

\theoremstyle{newstyle}

\newtheorem{theorem}{Theorem}
\newtheorem{lemma}{Lemma}
\newtheorem{proposition}{Proposition}
\newtheorem{corollary}{Corollary}

\theoremstyle{definition}
\newtheorem{example}{Example}
\newtheorem{definition}{Definition}

\theoremstyle{remark}
\newtheorem{remark}{Remark}
\newtheorem{claim}{Claim}

\sloppy

               \setlength{\belowcaptionskip}{-6pt}
        \setlength{\abovedisplayskip}{1mm}
        \setlength{\belowdisplayskip}{1mm}
        \setlength{\abovecaptionskip}{1mm}

        \title{Latency Analysis of Coded Computation Schemes over Wireless Networks}
        \author{\vspace{4mm}\large{Amirhossein Reisizadeh, Ramtin Pedarsani}\\
                 Department of Electrical and Computer Engineering\\
                 University of California, Santa Barbara\\
                 reisizadeh@ucsb.edu, ramtin@ece.ucsb.edu~~
                }
        
\maketitle

\begin{abstract}
Large-scale distributed computing systems face two major bottlenecks that limit their scalability: straggler delay caused by the variability of computation times at different worker nodes and communication bottlenecks caused by shuffling data across many nodes in the network. Recently, it has been shown that codes can provide significant gains in overcoming these bottlenecks. In particular, optimal coding schemes for minimizing latency in distributed computation of linear functions and mitigating the effect of stragglers was proposed in \cite{lee2016speeding} for a wired network, where the workers can simultaneously transmit messages to a master node without interference. 
In this paper, we focus on the problem of coded computation over a wireless master-worker setup with straggling workers, where only one worker can transmit the result of its local computation back to the master at a time. We consider 3 asymptotic regimes (determined by how the communication and computation times are scaled with the number of workers) and precisely characterize the total run-time of the distributed algorithm and optimum coding strategy in each regime. In particular, for the regime of practical interest where the computation and communication times of the distributed computing algorithm are comparable, we show that the total run-time approaches a simple lower bound that decouples computation and communication, and demonstrate that coded schemes are $\Theta(\log(n))$ times faster than uncoded schemes. 

\end{abstract}

\section{Introduction}

Modern large-scale computing systems are driven by scaling out computations across many small machines. While current distributed computing systems are mainly based on cloud computation with abundance of computational resources, there has been increasing interest in distributed computation over wireless networks where the computational nodes are mobile or wireless devices \cite{datla2012wireless,drolia2013case}. In a recent report by Google research, ``Federated Learning'' is proposed that enables mobile phones to collaboratively learn a shared prediction model while keeping all the training data on device, decoupling the ability to do machine learning from the need to store the data in the cloud \cite{googleresearch}. Further, wireless distributed computing enables many emerging mobile applications such as voice recognition, image processing, and virtual reality to get carried out without putting extra burden to the cloud. It can further improve the quality of service and latency of such applications.

Distributed computing systems encounter two major bottlenecks that limit their scalability: (i) Straggler delay bottleneck which is due to the latency in waiting for the slowest nodes to finish their computation tasks; (ii) Communication bottleneck which is due to the need in shuffling massive amounts of data over many nodes in the distributed algorithm. The traditional approach for mitigating these bottlenecks is to introduce computation redundancy in the form of task replicas. For example, replicating the straggling task on another available node is a common approach to mitigate the effect of stragglers \cite{zaharia2008improving}. However, recent results have proposed the use of ``coding schemes'' in introducing clever computational redundancy to deal with both stragglers and communication bottlenecks. 

\subsection{Related Works on Coded Computation}
The use of codes for minimizing latency in distributed computation of linear functions was introduced in \cite{lee2016speeding}. The key idea is to use erasure codes to inject redundancy such that the minimum latency is achieved by trading off the number of stragglers that the algorithm is robust to with redundancy factor in computation. In \cite{reisizadeh}, the authors consider coded computation over heterogeneous clusters, and proposed an asymptotically optimal coded algorithm for distributed matrix-vector multiplication. The use of product codes and polynomial codes for high-dimensional matrix multiplication over homogeneous clusters is proposed in \cite{lee-suh} and \cite{yu2017polynomial}, respectively. In a related work \cite{dutta2016short}, the authors propose the use of redundant short dot products to speed up distributed computation of linear transforms. Coded computing of the convolution of two long vectors distributedly in the presence of stragglers is proposed in \cite{dutta2017coded}. Coded computation of nonlinear functions over multicore setups is studied in \cite{multicore}. In \cite{tandon2016gradient}, the authors propose coding schemes for mitigating stragglers in distributed batch gradient computation. The idea of coded computation is utilized in  \cite{yang2017coding} for solving linear inverse problems in a parallelized implementation affected by stragglers.  

The use of codes for minimizing bandwidth in distributed computation was introduced in \cite{cmr,li2016fundamental}, and for coded data shuffling in distributed machine learning algorithms in \cite{lee2016speeding}. In \cite{li2017scalable}, the authors propose a scalable framework for minimizing the communication bandwidth in wireless distributed computing without considering stragglers delay. A unified coded framework for (wired) distributed computing with straggling servers is proposed in \cite{li2016unified}, by introducing a tradeoff between latency of computation and load of communication for some linear computation tasks. In a related word to coded data shuffling, \cite{attia2016information} studies the information theoretic limits of data shuffling in distributed learning. A pliate index coding approach is proposed in \cite{song2017pliable} for data shuffling.  

\subsection{Main Contribution}

All of the coded computation schemes for minimizing stragglers delay that have been proposed in the literature are for wired networks, where the results of computations of the workers can be received simultaneously by the master node, and there exists no interference between their channels. In this paper, we consider the problem of distributedly computing a linear computation task over a wireless master-worker setup with straggling workers, where \emph{only one worker can transmit the result of its local computation back to the master at a time.} This interference model is motivated by classical results in wireless networks such as \cite{tassiulas1992stability}, and can be generalized to the case where the dependency of the workers can be described by defining subsets that can be activated simultaneously. To characterize the total run-time of a distributed computing task, we consider two variables for the \emph{computation time} of each worker and for the \emph{communication time} of each locally computed result from the workers to the master. This separation of computation and communication time is also performed in \cite{multicore}. We propose optimal coding schemes for different regimes for speeding up distributed matrix-vector multiplication in wireless networks with straggling servers. Matrix-vector multiplication is a crucial computation in many distributed  machine learning algorithms such as logistic regression, reinforcement learning and gradient descent based algorithms. Wireless implementations of such algorithms can play a fundamental role in speeding up mobile data analytics and emerging mobile applications.

We now explain the main contribution of this paper. In the original coded computation paper \cite{lee2016speeding}, an optimal MDS code is proposed to minimize the latency of total computation run-time by robustifying the distributed algorithm to a few stragglers. Using an $(n,k)-$MDS code, it is easy to show that this total computation run-time is the $k$-th order statistics of $n$ i.i.d. random variables, assuming that the computation time of each worker is independent and identically distributed (i.i.d.); thus, one can optimize the rate of the code to find the best trade-off between computation redundancy and robustness to stragglers. In a wireless network, the total run-time is not simply determined by the $k$-th order statistics of $n$ i.i.d. random variables; instead, the system has \emph{memory} as the state of the computation is determined by the number of computations that are completed and the number of local computed results that are transmitted back to the master node. While exact characterization of the total run-time and optimal coding strategy seems to be intractable for arbitrary system parameters, in this paper we consider 3 asymptotic regimes (determined by how the communication and computation times are scaled with the number of workers) and precisely characterize the total run-time of the distributed algorithm and optimum coding strategy in each regime. In particular, for the regime of practical interest where the computation and communication times of the distributed computing algorithm are comparable, we show that the total run-time approaches a simple lower bound that decouples computation and communication, and demonstrate that coded schemes are $\Theta(\log(n))$ times faster than uncoded schemes. 

\subsection{Notation}

In this paper, we use bold face small and capital letters for vectors and matrices respectively, capital letters for random variables, and lower cases for constants. For $i\in \mathbb{N}$, we denote by $[i]$ the set $\{1,2,\cdots,i\}$. For non-negative functions $f$ and $g$, we denote $f(n)=\mathcal{O}(g(n))$ if there exists $n_0\in \mathbb{N}$ and $c>0$ such that $f(n) \leq c g(n)$ for $n\geq n_0$; $f(n)=o(g(n))$ if $f(n)/g(n) \rightarrow 0$ as $n \to \infty$. Moreover, $f(n)=\omega(g(n))$ if and only if $g(n)=o(f(n))$ and $f(n)=\Theta(g(n))$ if and only if $f(n)=\mathcal{O}(g(n))$ and $g(n)=\mathcal{O}(f(n))$.

\section{Problem Statement}

\begin{figure}
 \begin{center}
 \includegraphics[width=7cm]{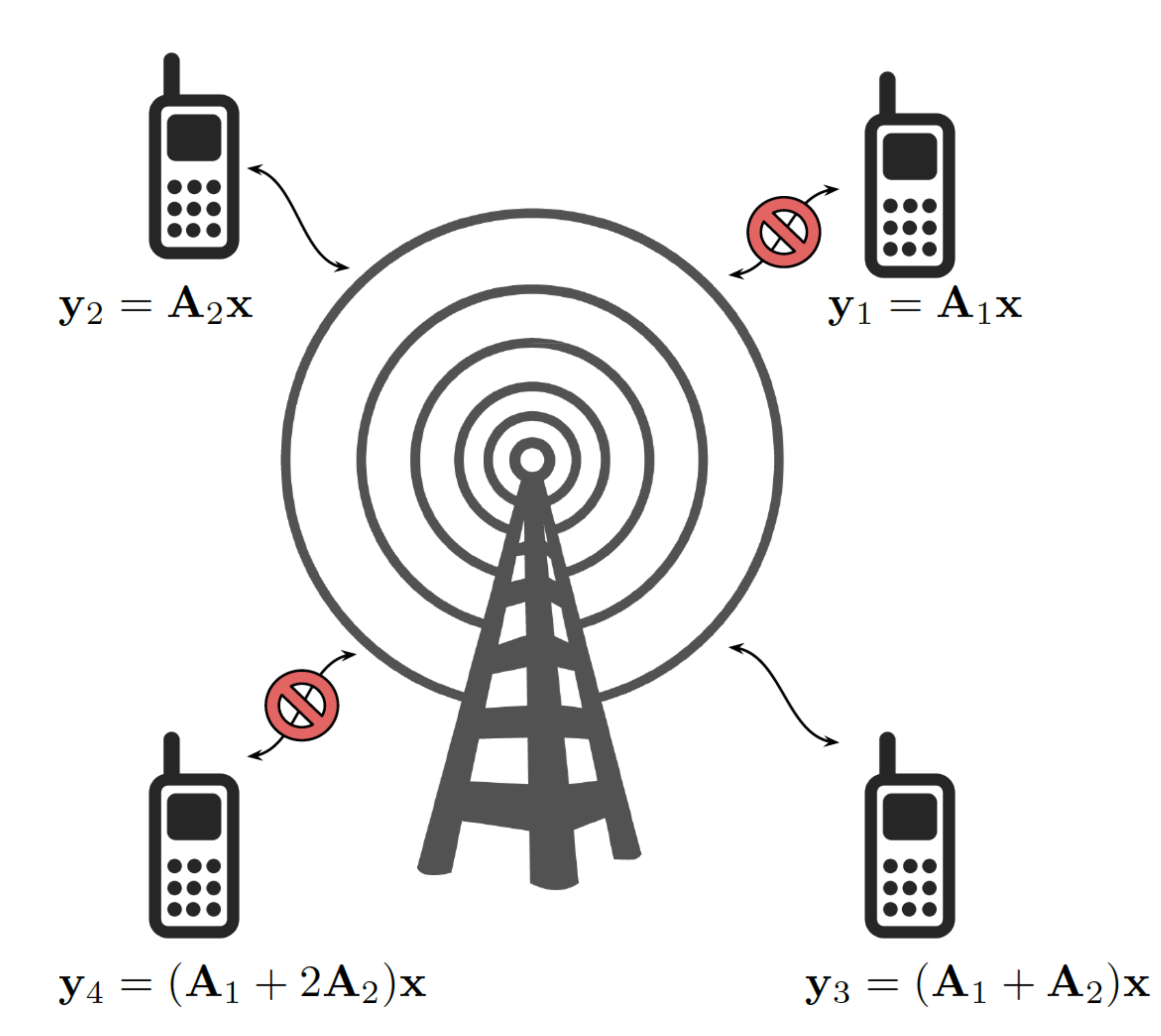}
    \caption{\small{\textbf{Problem formulation.} We consider the problem of matrix-vector multiplication over a wireless master-worker setup, where workers on local computation of coded data and transmit the results of their computations back to the master node. In this model, only one worker can communicate to the master node at a time.}}
    \label{fig:setup}
    \end{center}
\end{figure}
  
Consider a master-worker system consisting of a master node and $n$ computing worker nodes. The goal is to compute the matrix-vector multiplication $\mathbf{y}=\mathbf{Ax}$ distributedly across the worker nodes, where $\mathbf{A}\in\mathbb{R}^{r\times m}$ and $\mathbf{x}\in\mathbb{R}^{ m}$. Every worker has access to the vector $\mathbf{x}=(x_1,\cdots,x_m)^\top$, where each $x_i$ can be represented by $d$ bits, i.e., $x_i\in \mathbb{F}_2^d$, for $i\in [m]$. Thus, there are overall $r=\Theta(n)$ inner product computations required to be carried out by the worker nodes. 
The master node assigns computing the matrix-vector multiplication $\mathbf{y}_i=\widetilde{\mathbf{A}}_i\mathbf{x}$ to worker $i$ for $i \in [n]$ where $\widetilde{\mathbf{A}}_i$'s are functions of $\mathbf{A}$. Due to memory and computation capability constraint on local computing nodes, matrices $\widetilde{\mathbf{A}}_i$'s are usually much smaller in size than the original matrix $\mathbf{A}$. Each worker performs its partial computation and sends the result back to the master node. The master node aggregates the partial computations $\mathbf{y}_i$'s and retrieves $\mathbf{y}$ from a recoverable set of local computations. 

Computation time of each worker node for one inner product is modeled as a shifted-exponential random variable with parameters $(a,\mu)$, i.e.
\begin{equation}
\Prob(T_i\leq t) = 1-e^{-\mu(t-a)},
\end{equation}
for $t\geq a$, where $T_i$ denotes the computation latency of worker node $i\in [n]$ for computing an inner product of two vectors of length $m$. This computation time model is motivated by the distribution model proposed by authors in \cite{liang13} for latency in querying data files from cloud storage systems. 
For the sake of notation simplicity, we denote such a random variable by $T_i \sim a+\exp(\mu)$.
 Workers can communicate with the master node via wireless links with the same rate. We consider a MAC layer model where only one worker can transmit data to the master node at a time. Thus, every worker that finishes its computation, waits till the channel becomes idle and initiates the transmission, and priority is given to nodes that have earlier computation completion times. In other words, communications are carried out in the same order that their corresponding computations are executed. Let $t_{\text{1-cmm}}$ denote the communication time for the result of a single inner product. Furthermore, we assume the communication phase of every worker node is without preemption, i.e. a worker node transmits its result in one round of communication. 
 Since only one worker is able to transmit results to the master node at a time, one needs to reexamine the total execution time (latency) of the distributed computation, $T_{\text{tot}}$, which is defined as the total time of computation and communication to carry out the complete multiplication. As we will see, precisely characterizing the total run-time $T_{\text{tot}}$ is intractable for arbitrary set of problem parameters and requires combinatorial analysis. Instead, we provide asymptotic characterizations of the total latency for three different regimes. We consider coded and uncoded scenarios for a system of one master and $n$ workers and evaluate the total execution time for each scenario.

\section{Main Results: Wireless Coded Computation}
In this section, we investigate the total execution time of coded computation schemes over wireless networks in different regimes of parameters. We first illustrate the key idea of coded computation through a simple example.

\begin{example}
Consider a computation system consisting of one master node and $n=4$ worker nodes as depicted in Figure \ref{fig:setup}. To robustify the computation to stragglers, we carry out a $(4,2)$--MDS coded computation as follows. The master node divides the large matrix $\mathbf{A}$ to two submatrices $\mathbf{A}_1$ and $\mathbf{A}_2$ of the same size:
\begin{equation}
    \mathbf{A}=
    \begin{bmatrix}
    \mathbf{A}_1   \\
    \mathbf{A}_2
\end{bmatrix}.
\end{equation}
Then, the local computations assigned to each worker node is as follows
\begin{align}
\mathbf{y}_1 &=\widetilde{\mathbf{A}}_1\mathbf{x}=\mathbf{A}_1\mathbf{x},\nonumber\\
\mathbf{y}_2 &=\widetilde{\mathbf{A}}_2\mathbf{x}=\mathbf{A}_2\mathbf{x},\nonumber\\
\mathbf{y}_3 &=\widetilde{\mathbf{A}}_3\mathbf{x}=(\mathbf{A}_1+\mathbf{A}_2)\mathbf{x},\nonumber\\
\mathbf{y}_4 &=\widetilde{\mathbf{A}}_4\mathbf{x}=(\mathbf{A}_1+2\mathbf{A}_2)\mathbf{x}.
\end{align}
Clearly, results of any $k=2$ workers (e.g. workers 2 and 3 in Figure \ref{fig:setup}) compose a recoverable set to retrieve 
\begin{equation}
    \mathbf{y}=\mathbf{A}\mathbf{x}=
    \begin{bmatrix}
    \mathbf{A}_1 \mathbf{x}  \\
    \mathbf{A}_2\mathbf{x}
\end{bmatrix}.
\end{equation}
Therefore, the master node only waits for the two fastest workers and recovers the result then.
\end{example}

There are totally $r$ inner products to be computed via $n$ nodes distributedly. We employ an $(n,k)-$MDS code to perform the computation task, i.e. the computation result of any $k$ worker nodes forms a decodable set of inner products. Therefore, every node is assigned $r/k$ inner products to be computed. Each assigned vector to the worker nodes is a random linear combination of the rows of the matrix $\mathbf{A}$, thus the master node assigns $r/k$ random linear combinations of the rows of $\mathbf{A}$ to each worker. More precisely, the computation matrix assigned to worker $i$ is $\widetilde{\mathbf{A}}_i=\mathbf{S}_i\mathbf{A}$ where $\mathbf{S}_i\in \mathbb{R} ^{r/k\times r}$ is the coding matrix for worker $i$. To establish a random linear code, the entries of $\mathbf{S}_i$ are i.i.d. $\mathcal{N}(0,1)$. Worker $i$ computes $\mathbf{y}_i=\widetilde{\mathbf{A}}_i\mathbf{x}$ and sends back the result to the master node. Upon receiving $r$ inner products, the master node can retrieve the computation $\mathbf{y}=\mathbf{A}\mathbf{x}$ with probability 1. The master node aggregates the results in the form $\mathbf{z}=\mathbf{S}_{(r)}\mathbf{A}\mathbf{x}$,  where $\mathbf{S}_{(r)}\in \mathbb{R}^{r\times r}$ is the aggregated coding matrices which is full-rank with probability 1. Therefore, the master node can recover $\mathbf{y}=\mathbf{A}\mathbf{x}=\mathbf{S}_{(r)}^{-1}\mathbf{z}$. Similarly, one can construct an MDS code to ensure any $k$ out of $n$ responses of the workers suffice for recovering $\mathbf{y}$.

Since each worker node computes $r/k$ inner products, adopting the computation time model in \cite{lee2016speeding}, the random variable denoting the computation time for worker node $i$ can be represented as a constant shift $t_0=ar/k$ added to an exponential term $T_i$ with rate $\mu k/r$, i.e.
\begin{equation}
\Prob(t_0+T_i\leq t) = 1-e^{-\frac{\mu k}{r}(t-t_0)},
\end{equation}
for $t\geq t_0$ and $i \in [n]$, where $T_i$'s are i.i.d.  (see figure \ref{fig:timing}).\footnote{For ease of notation, from now on we represent the random computation time of worker $i$ as the sum of a constant term denoted by $t_0$ and a variable term with exponential distribution denoted by $T_i$.} We denote the ordered sequence of computation times by $\{t_0+T_{(i)}\}_{i=1}^n$  where $t_0+T_{(i)}$ denotes the $i$-th smallest computation time ($i$-th order statistic), i.e. $T_{(1)} \leq T_{(2)} \leq \cdots \leq T_{(n)}$. Moreover, we define the differences in computation times (differential times) as follows 
\begin{align}
D_1 &= T_{(1)},\nonumber\\
D_i &= T_{(i)} - T_{(i-1)},
\end{align}
for $i=2,\cdots,n$. Let $\alpha=\frac{r}{\mu k}$. Differential times are exponential and mutually independent random variables with the following distribution: $D_i \sim \exp(\frac{1}{\alpha}(n-i+1))$. Furthermore, communication time for each worker is $t_{\text{cmm}}=\frac{r}{k}t_{\text{1-cmm}}$ that is a constant term. The master node needs to wait for the results of the fastest $k$ worker nodes. Trivial lower and upper bounds on the total execution time can be derived as follows
\begin{equation}
    t_0+T_{(k)} + t_{\text{cmm}} \leq T_{\text{tot}} \leq t_0+ T_{(k)} + kt_{\text{cmm}},
\end{equation}
almost surly. Computation times $T_{(i)}$'s are i.i.d exponential with rate $\mu k/r$, therefore the expected value of the $i$-th order statistics is $\Expc[T_{(i)}]=\frac{H_n - H_{n-i}}{\mu k/r}$ where $H_i=1+\frac{1}{2}+\cdots+ \frac{1}{i}$ (see Lemma \ref{lemma:expc_approx}). 

One can easily find lower and upper bounds on the total expected run-time as follows: 
\begin{equation}\label{eq:bounds}
    t_0+\alpha (H_n - H_{n-k}) + t_{\text{cmm}} \leq \Expc[T_{\text{tot}} ]\leq t_0+ \alpha (H_n - H_{n-k}) + kt_{\text{cmm}}.
\end{equation}
Clearly, if the channel is idle at the time of the completion of $k$-th computation, one finds that the total run-time is the left-hand side of \eqref{eq:bounds} which implies the lower bound. Further, if one considers a (higher-latency) protocol where all the communications are performed after the $k$-th computation, the upper bound will be derived.

To illustrate the challenge of exactly characterizing the total run-time, we provide a simple example describing the dynamics of the problem.

\begin{example}
Assume $a=\frac{1}{\mu}=1$ second, $n=5$, $k=3$, $r=5$, and $t_{\text{cmm}}=0.2$ seconds. From a random realization, we get $T_{(1)}=0.1138$, $T_{(2)}=0.2725$, $T_{(3)}=0.6458$, $T_{(4)}=0.7033$, and $T_{(5)}=5.5538$ all in seconds. Figure \ref{fig:timing} depicts the time digram of the experiment. Thus, one observes that the communication of the second result cannot be carried out at $t_0 + T_{(2)}$ since the channel is busy. However, the channel is idle at $t_0 + T_{(3)}$. Thus, the overall run-time is the same as the lower bound in \eqref{eq:bounds}.

\begin{figure}[h]
\begin{center}
\includegraphics[width=13cm]{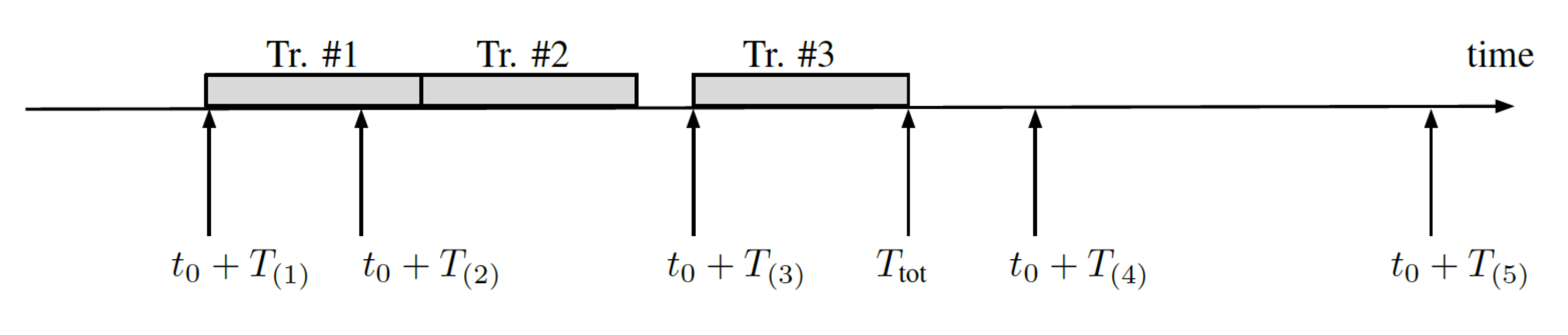}
    \caption{\small{\textbf{Computation and communication time diagram.} This time diagram shows the computation times of each worker and when the result of their computation is transmitted to the master node. One can observe that if the channel is busy transmitting the message of another worker, the communication of other completed computation results have to wait for the channel to get idle.}}
    \label{fig:timing}
    \end{center}
\end{figure}

\end{example}

For the $(n,k)-$MDS coded computation, $k$ is the design parameter and is picked in the order of $n$: $k=\Theta(n)$. Thus, $\alpha=\Theta(1)$. In the following lemma, we show that as $j$ gets large, the random variable $T_{(j)}$ is well-concentrated around its expected value.
\begin{lemma}\label{lemma:expc_approx}
For $j=\Theta(n)$, we have $| T_{(j)} - \mathbb{E}[T_{(j)}] | \leq o(1)$ with probability $1-o(1)$.
\end{lemma}
\begin{proof}
We can write 
\begin{equation}
T_{(j)}=\sum_{i=1}^{j} D_i,
\end{equation}
where $D_i\sim \exp(\alpha^{-1}(n-i+1))$ for $i=1, \dots, n$, and they are independent. Therefore, the expected value of $T_{(j)}$ can be written as 
\begin{align}
\Expc[T_{(j)}]&= \sum_{i=1}^{j} \Expc[D_i]\nonumber\\
&=  \sum_{i=1}^{j} \frac{\alpha}{n-i+1}\nonumber\\
&= \alpha (H_n - H_{n-j}) \nonumber\\
&= \alpha \log (\frac{n}{n-j}) + \mathcal{O}(\frac{1}{n}).
\end{align}
Harmonic series $H_n=1+\frac{1}{2}+\cdots+ \frac{1}{n}$ can be approximated as $H_n=\log n + \gamma + \mathcal{O}(\frac{1}{n})$ where $ \gamma \approx 0.72156649$ is the Euler-Mascheroni constant. Moreover, for $j=\Theta(n)$, we may write the variance of $T_{(j)}$ as
\begin{align}
\var[T_{(j)}]&= \sum_{i=1}^{j} \var[D_i]\nonumber\\
&=  \sum_{i=1}^{j} \frac{\alpha^2}{(n-i+1)^2}\nonumber\\
&= \alpha^2  \Theta(\frac{1}{n})\nonumber\\
&=  \Theta(\frac{1}{n}).
\end{align}
By Chebyshev's inequality, 
\begin{align}
\Prob \big(| T_{(j)} - \mathbb{E}[T_{(j)}] |\leq \epsilon \big) \geq 1-\frac{\var[T_{(j)}]}{\epsilon^2}=1-o(1),
\end{align}
for $\epsilon=\Theta(\frac{\log n}{n})$. 
\end{proof}
This lemma indicates that for large enough $j$, we can approximate $T_{(j)}$ by its expected value, i.e. with probability approaching 1, we have $T_{(j)}=\alpha \log (\frac{n}{n-j}) \pm o(1)$.

\begin{lemma}\label{lemma:expc_approx2}
For $j=o(n)$, we have $T_{(j)} = o(1)$, with probability $1-o(1)$.
\end{lemma}
\begin{proof}
We can write the expected value of $T_{(j)}$ as 
\begin{align}
\Expc[T_{(j)}]&= \alpha (H_n - H_{n-j}) \nonumber\\
&= \alpha \log (\frac{n}{n-j}) + \mathcal{O}(\frac{1}{n})\nonumber\\
&=\alpha \log (1+\frac{j}{n-j}) + \mathcal{O}(\frac{1}{n})\nonumber\\
&= \alpha \frac{j}{n-j} + o(1)\nonumber\\
&= o(1).
\end{align}
Following the same argument in the proof of Lemma \ref{lemma:expc_approx}, the variance of $T_{(j)}$ for $j=o(n)$ can be evaluated as 
\begin{align}
\var[T_{(j)}]&= \sum_{i=1}^{j} \var[D_i]\nonumber\\
&= \alpha^2 \sum_{i=1}^{j} \frac{1}{(n-i+1)^2}\nonumber\\
&\leq \alpha^2  \frac{j}{(n-j+1)^2}\nonumber\\
&=  \mathcal{O}(\frac{j}{n^2})\nonumber\\
&= o(\frac{1}{n}).
\end{align}
By Chebyshev's inequality, 
\begin{align}
\Prob(| T_{(j)} - \mathbb{E}[T_{(j)}] |\leq \epsilon ) \geq 1-\frac{\var[T_{(j)}]}{\epsilon^2}=1-o(1),
\end{align}
for $\epsilon=\Theta(\frac{1}{\sqrt{n}})$. Therefore, we conclude that for $j=o(k)$, we have $T_{(j)}=\mathbb{E}[T_{(j)}] \pm o(1)=o(1)$, with probability $1-o(1)$.
\end{proof}

As discussed before, characterizing the exact execution time for arbitrary parameters seems to be intractable for arbitrary parameters. Instead, we consider three asymptotic regimes and evaluate the overall run-time in the following.
\subsection{Regime I}
In this regime, at a high level the total communication time is negligible compared to computation time. More precisely, we consider the regime where $t_{\text{1-cmm}}=o(\frac{1}{n})$. For $k=\Theta(n)$, we have  $t_{\text{cmm}}=o(\frac{1}{n})$. In this regime, the communication time for each worker node is small enough such that almost all of the results can be communicated during the computation phase. Thus, the problem gets reduced to the wired coded computation problem in \cite{lee2016speeding} as $n$ gets large. The following theorem precisely states this fact.
\begin{theorem}
With high probability, $k$ transmissions are completed by time \normalfont $t_0+T_{(k)}+t_{\text{cmm}}$.
\end{theorem}

\begin{proof}

Let $Q\leq k$ denote the greatest random index for which the channel is idle at time $t_0 + T_{(Q)}$. First, assume that $Q=o (n)$ with probability $1 - o(1)$. Given the realization $Q = q = o(n)$,  by Lemma \ref{lemma:expc_approx2}, with high probability,
\begin{equation}
T_{(q)} = o(1).
\end{equation}
Therefore, $T_{(k)}=T_{(q)}+(k-q)t_{\text{cmm}}=o(1)$, which is in contradiction to Lemma \ref{lemma:expc_approx}, since $T_{(k)}=\Theta(1)$. Secondly, assume that $Q = q = \Theta(n)<k$ with probability $1 - o(1)$. From Lemma \ref{lemma:expc_approx} we can write 
\begin{align}
T_{(k)}-T_{(q)} &\geq \alpha \log (\frac{n}{n-k}) - \alpha \log (\frac{n}{n-q}) - o(1) \nonumber\\
&=\alpha \log (\frac{n-q}{n-k})  -o(1)\nonumber\\
&= \Theta(1).\nonumber
\end{align}
On the other hand, since the channel is not idle in the interval $[t_0+T_{(q)},t_0+T_{(k)}]$ one obtains 
$$
T_{(k)}-T_{(q)} = (k-q)t_{\text{cmm}} = o(1)
$$ 
that is a contradiction. Thus, $Q = q = k$ with probability $1 - o(1)$. This implies that all the $k-1$ computations are transmitted by time  $t_0+T_{(k)}$ with high probability. The last transmission would occur right after the corresponding computation is finished. Therefore, the lower bound $t_0+T_{(k)}+t_{\text{cmm}}$ in \eqref{eq:bounds} is achieved and $T^{\text{R}_{\text{I}}}_{\text{coded}}=t_0+T_{(k)}+t_{\text{cmm}}$ with probability $1-o(1)$, where $T^{\text{R}_{\text{I}}}_{\text{coded}}$ denotes the run-time corresponding to the $(n,k)-$MDS coded scheme performing in regime I. 
\end{proof}

\subsection{Regime II}
In this regime, computation time is negligible compared to communication time, i.e. $t_{\text{1-cmm}}=\omega(\frac{1}{n})$ which implies $t_{\text{cmm}}=\omega(\frac{1}{n})$ for $k=\Theta(n)$.  

The following theorem states that the computation times are small enough such that most of the communications have to occur after the last computation is finished.
\begin{theorem}
With probability $1-o(1)$, at most $o(n)$ transmissions are completed by time $t_0 +T_{(k)}$.
\end{theorem}
\begin{proof}
Let $Q\leq k$ denote the random variable denoting the greatest index for which the channel is idle at time $t_0 + T_{(q)}$. First, assume that $Q=\Theta (n)$ with probability $1 - o(1)$. Given the realization $Q = q = \Theta(n)$,  by Lemma \ref{lemma:expc_approx}, with high probability,
\begin{equation}\label{eq:16}
T_{(q)} = \alpha \log (\frac{n}{n-q}) \pm o(1) = \Theta(1).
\end{equation}
On the other hand, all the $q-1$ communications are finished by time $T_{(q)}$. Therefore, 
\begin{equation}
T_{(q)} \geq (q-1) t_{\text{cmm}}= \omega(1),
\end{equation}
which is in contradiction to (\ref{eq:16}). Hence, $Q=o(n)$ with probability $1 - o(1)$, and given that $Q = q = o(n)$ the number of transmissions completed by time $t_0+T_{(k)}$ can be written as 
\begin{align}
\text{\# completed transmissions by } t_0+T_{(k)}&=q-1+\frac{T_{(k)}-T_{(q)}}{t_{\text{cmm}}}\\ &\leq q+\frac{T_{(k)}}{t_{\text{cmm}}} \nonumber\\
&= o(n) + \frac{\Theta(1)}{\omega(\frac{1}{n})}\nonumber\\
&= o(n).
\end{align}
Therefore, $T^{\text{R}_{\text{II}}}_{\text{coded}}\geq t_0+T_{(k)}+(k-o(n))t_{\text{cmm}}$ with probability $1-o(1)$, where $T^{\text{R}_{\text{II}}}_{\text{coded}}$ denotes the run-time corresponding to the $(n,k)-$MDS coded scheme performing in regime II.
\end{proof}

\subsection{Regime III}
The third regime is the regime of interest where communication and computation times are comparable, i.e. $t_{\text{1-cmm}}=\Theta(\frac{1}{n})$. Without loss of generality, we assume that $t_{\text{1-cmm}}=\frac{k}{rn}$ and therefore $t_{\text{cmm}}=\frac{1}{n}$. 

We now define a new variable $p$ as follows. Let $p$ be the smallest integer such that
\begin{equation}
    \sum_{i=1}^{p} \mathbb{E}[D_i] \geq (p-1)t_{\text{cmm}}, 
\end{equation}
that is
\begin{equation}\label{eq:pdef}
    \sum_{i=1}^{p} \frac{\alpha}{n-i+1} \geq \frac{p-1}{n}. 
\end{equation}
Thus, for $j<p$, we have 
\begin{equation}
    \sum_{i=1}^{j} \frac{\alpha}{n-i+1} < \frac{j-1}{n},
\end{equation}
and for $j>p$, 
\begin{equation}
    \sum_{i=1}^{j} \frac{\alpha}{n-i+1} > \frac{j-1}{n}. 
\end{equation}
Moreover, 
\begin{equation}\label{eq:48}
    \sum_{i=1}^{p} \frac{\alpha}{n-i+1} < \frac{p-1}{n}+o(1). 
\end{equation}

From definition of the index $p$, it is easy to check that $p=\Theta(n)$.

To illustrate this definition, let us first discuss a naive approximation of the differential times as follows. We approximate the differential times with deterministic variables equal to their expected values, i.e.
\begin{equation}
    \widetilde{D}_i = \Expc[D_i]=\frac{\alpha}{n-i+1}.
\end{equation}
This assumption implies that the inter-arrival times are increasing and therefore the first $p-1$ computations are pipe-lined or the channel does not become idle in the first $p-1$ transmissions (see figure \ref{fig:timing2}), i.e. $\sum_{i=1}^{p} \widetilde{D}_i \geq (p-1)t_{\text{cmm}}$. Thus, the constant $p$ approximates the first index for which the channel is idle when the $p$-th computation is performed; thus, the communication can be initiated right after the computation result is ready for transmission. 

Now we get back to the random computation times and first evaluate the number of transmissions communicated by time $t_0+T_{(p)}$.

\begin{figure}[h]
\begin{center}
\includegraphics[width=13cm]{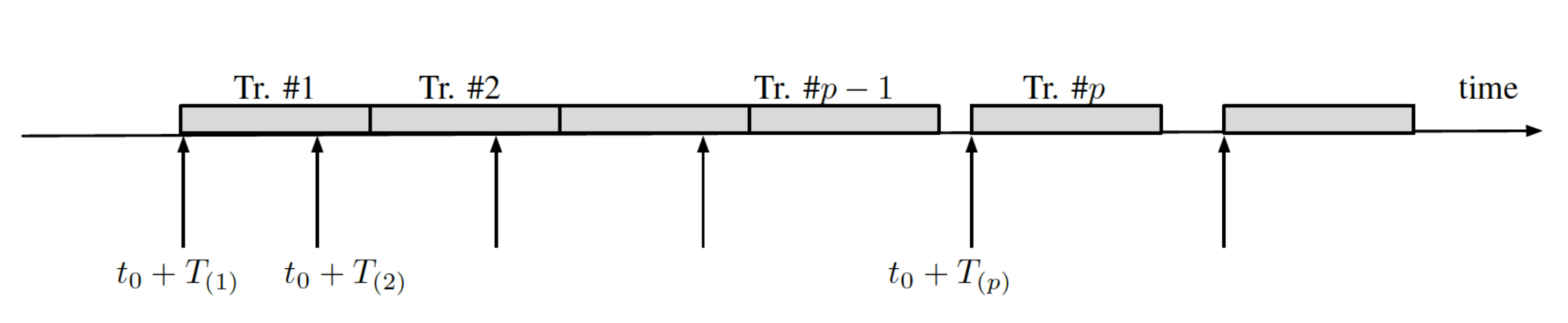}
    \caption{\small{\textbf{Computation and communication timing for approximate computation times.} Using the approximation that the computation times are close to their expected values, we obtain $p$ which approximates the first index for which the channel is idle when the $p$-th computation is performed. }}
    \label{fig:timing2}
    \end{center}
\end{figure}

\begin{lemma}\label{lemma:main1}
With high probability, at least $p-o(n)$ transmissions are completed by time $t_0+T_{(p)}$.
\end{lemma}
\begin{proof}
We define $Q$ as the largest index less than $p$ for which the channel is busy from $t_0+T_{(Q)}$ to $t_0+T_{(p)}$, and $q$ as its realization. We consider the following two possible regimes for $Q$.
\begin{itemize}
\item Case 1: $Q = q=o(n)$.
Given the assumption $q=o(n)$ and using Lemma \ref{lemma:expc_approx2}, we can write $T_{(q)}=\mathbb{E}[T_{(q)}] \pm o(1)=o(1)$, with probability $1-o(1)$. Thus, using Lemma \ref{lemma:expc_approx}, the number of completed transmissions up to time $t_0+T_{(p)}$ is 
\begin{align}
\text{\# completed transmissions by } t_0+T_{(p)}&= q-1+n(T_{(p)}-T_{(q)})\nonumber\\
&\geq nT_{(p)} - o(n)\nonumber\\
&\geq n \alpha (H_n - H_{n-p}) - o(n)-o(1)\nonumber\\
& = n\alpha (H_{n-1} - H_{n-p-1}) - o(n)\nonumber\\
& \geq p-o(n).
\end{align}
\item Case 2: $Q = q=\Theta(n)$.
Using Lemma \ref{lemma:expc_approx}, we have $T_{(q)}=\alpha \log (\frac{n}{n-q}) \pm o(1)$ and $T_{(p)}=\alpha \log (\frac{n}{n-p}) \pm o(1)$ with probability $1-o(1)$. Therefore, with probability $1-o(1)$, 
\begin{align}\label{eq:pq}
T_{(p)}-T_{(q)} &\geq \alpha \log (\frac{n}{n-p}) - \alpha \log (\frac{n}{n-q}) - o(1) \nonumber\\ 
&\geq \frac{p-q}{n} -o(1),
\end{align}
where \eqref{eq:pq} is due to the definition of $p$ in \eqref{eq:pdef} and the fact that $q \leq p$.
Therefore, by Lemma \ref{lemma:expc_approx}, $T_{(p)}$ and $T_{(q)}$ are concentrated around their expected values, and the number of completed transmissions up to time $t_0+T_{(p)}$ is 
\begin{align}
\text{\# completed transmissions by } t_0+T_{(p)}&= q-1+n(T_{(p)}-T_{(q)})\nonumber\\
&\geq q-1 + p-q- n\cdot o(1)\nonumber\\
&=p-o(n).
\end{align}
 Therefore, $p-o(n)$ transmissions are completed by time $t_0+T_{(p)}$ with probability $1-o(1)$.
 \end{itemize}
\end{proof}

\begin{lemma}\label{lemma:main2}
With high probability, at least $k-p-o(n)$ transmissions are completed from $t_0+T_{(p)}$ to $t_0+T_{(k)}$.
\end{lemma}
\begin{proof}
First, we note that by Lemma \ref{lemma:expc_approx}, $T_{(p)}$ and $T_{(k)}$ are well concentrated around their expected values. Thus, together with (\ref{eq:48}), 
\begin{align}
T_{(k)}-T_{(p)} &\geq \alpha \log (\frac{n}{n-k}) - \alpha \log (\frac{n}{n-p}) - o(1) \nonumber\\
&\geq \frac{k-p}{n} -o(1),
\end{align}
with probability $1-o(1)$.

Now in contrary, suppose that the number of communications from $t_0+T_{(p)}$ to $t_0+T_{(k)}$ is smaller than $\beta(k-p)$ for some constant $\beta < 1$ with a positive $\Theta(1)$ probability. Now let $T_{(j)}$ be the last time that the channel is idle in this period. Consider the following three cases:
\begin{itemize}
    \item Case 1: If $j = p$, then we reach a contradiction since $n(T_{(k)} - T_{(p)}) \geq  k-p - o(n)$ with high probability.
\item Case 2: If $p < j < \gamma k$ for some $\gamma <1$, then $n(T_{(k)} - T_{(j)}) \geq  k-j-   o(n)$ with probability $1 - o(1)$ and also $j-p$ communications have been done by time $t_0+T_{(j)}$ since the channel was idle, which implies a contradiction.
\item Case 3: If $j = k - o(n)$, since the channel was idle at that time, at least $j-p=k - p-o(n)$ communications have already been completed by $t_0+T_{(k)}$ that is a contradiction.
\end{itemize}
Therefore, the claim is concluded.

\end{proof}

Putting Lemmas \ref{lemma:main1} and \ref{lemma:main2} together, we conclude that at least $k-o(n)$ transmissions are completed by time $t_0+T_{(k)}$ with probability approaching 1. That is, $T^{\text{R}_{\text{III}}}_{\text{coded}}\leq t_0+T_{(k)}+o(n)t_{\text{cmm}}$ with probability $1-o(1)$, where $T^{\text{R}_{\text{III}}}_{\text{coded}}$ denotes the run-time corresponding to the $(n,k)-$MDS coded scheme performing in regime III.

\begin{theorem}\label{theorem:main}
The expected total run-time for an $(n,k)-$MDS coded scheme performing in Regime III is 
\begin{equation}
    \Expc[T^{\text{R}_{\text{III}}}_{\text{coded}}]=t_0+\alpha (H_n - H_{n-k})+o(1).
\end{equation}
\end{theorem}

\begin{proof}
From the lower bound in (\ref{eq:bounds}), $\Expc[T^{\text{R}_{\text{III}}}_{\text{coded}}] \geq t_0+\alpha (H_n - H_{n-k}) + o(1)$. Moreover, we proved that $T^{\text{R}_{\text{III}}}_{\text{coded}}\leq t_0+T_{(k)}+o(n)t_{\text{cmm}}$ with probability $1-o(1)$. Further, we have the upper bound $T^{\text{R}_{\text{III}}}_{\text{coded}}\leq t_0+T_{(k)}+kt_{\text{cmm}}$ in (\ref{eq:bounds}). Thus we can write 
\begin{align}
\Expc[T^{\text{R}_{\text{III}}}_{\text{coded}}] =& \Prob(T^{\text{R}_{\text{III}}}_{\text{coded}}\leq t_0+T_{(k)}+o(n)t_{\text{cmm}}) \Expc[T^{\text{R}_{\text{III}}}_{\text{coded}}|T^{\text{R}_{\text{III}}}_{\text{coded}}\leq t_0+T_{(k)}+o(n)t_{\text{cmm}}] \nonumber\\
&+ \Prob(T^{\text{R}_{\text{III}}}_{\text{coded}} > t_0+T_{(k)}+o(n)t_{\text{cmm}}) \Expc[T^{\text{R}_{\text{III}}}_{\text{coded}}|T^{\text{R}_{\text{III}}}_{\text{coded}}> t_0+T_{(k)}+o(n)t_{\text{cmm}}]\nonumber\\
\leq & (1-o(1))(t_0+T_{(k)}+o(n)t_{\text{cmm}})+o(1)(t_0+T_{(k)}+kt_{\text{cmm}})\nonumber\\
=&t_0+\alpha (H_n - H_{n-k})+o(1).
\end{align}
Therefore, $\Expc[T^{\text{R}_{\text{III}}}_{\text{coded}}]=t_0+\alpha (H_n - H_{n-k})+o(1)$.
\end{proof}

According to Theorem \ref{theorem:main}, the lower bound in (\ref{eq:bounds}) is achieved within a $o(1)$ additive factor for regime III. Note that one can optimize the rate of the code and find the best $k$ to minimize the expected total run-time numerically similar to \cite{lee2016speeding}. Moreover, note that for $k=\Theta(n)$, we have 
\begin{equation}
    \Expc[T^{\text{R}_{\text{III}}}_{\text{coded}}]=\Theta(1).
\end{equation}

In the next section, we will evaluate the total run-time for the uncoded scheme and compare it with the one corresponding to the coded scenario.

\section{Uncoded Computation over Wireless Networks}

In the uncoded scheme, the workload (total $r$ inner products) is evenly distributed among $n$ workers. Thus, each node computes $r/n$ inner products. The master node has to wait for all the workers to finish their computations and send the results back to the master node. The following theorem compares the total run-time corresponding to coded and uncoded schemes and demonstrates how much gain one achieves by employing the proper coding strategy.
\begin{theorem}
In Regimes I and III, coded computation is $\Theta(\log n)$ times faster than uncoded computation, i.e. 
\begin{equation}
    \frac{\Expc[T_{\text{uncoded}}]}{\Expc[T^{\text{R}_{\text{I}}}_{\text{coded}}]}=\Theta(\log n), \qquad \frac{\Expc[T_{\text{uncoded}}]}{\Expc[T^{\text{R}_{\text{III}}}_{\text{coded}}]}=\Theta(\log n).
\end{equation}
\end{theorem}
\begin{proof}

In the uncoded scenario, the master node needs to wait for the result of all the $n$ computations. Therefore, with probability 1,
\begin{equation}
    \frac{ar}{n}+T_{(n)}+\frac{1}{n} \leq T_{\text{uncoded}} \leq \frac{ar}{n}+T_{(n)}+1,
\end{equation}
which implies 
\begin{equation}
    \frac{ar}{n}+\Expc[T_{(n)}]+\frac{1}{n} \leq \Expc[T_{\text{uncoded}}] \leq \frac{ar}{n}+\Expc[T_{(n)}]+1.
\end{equation}
Thus, 
\begin{equation}
    \frac{ar}{n}+\alpha H_n +\frac{1}{n} \leq \Expc[T_{\text{uncoded}}] \leq  \frac{ar}{n}+\alpha H_n +1,
\end{equation}
which concludes
\begin{equation}\label{eq:uncoded}
    \Expc[T_{\text{uncoded}}]=\Theta(H_n)=\Theta(\log n).
\end{equation}
According to Theorem \ref{theorem:main}, for a $(n,k)-$MDS coded scheme with $k=\Theta(n)$, we have $\Expc[T^{\text{R}_{\text{III}}}_{\text{coded}}]=\Theta(1)$. Considering the trivial lower and upper bounds in (\ref{eq:bounds}) implies that $\Expc[T^{\text{R}_{\text{I}}}_{\text{coded}}]=t_0+\Expc[T_{(k)}]+o(1)=\Theta(1)$. Together with (\ref{eq:uncoded}), the claim is concluded.
\end{proof}

\section{Conclusion and Future Work}

We considered the problem of coded computation over a wireless network with master-worker setup and straggling workers. In this network, only one worker can transmit message to the master node at a time. We proposed the use of optimal MDS-coded schemes to minimize the total run-time of the distributed computation algorithm. In particular, while the exact characterization of total run-time is not tractable, we considered 3 asymptotic regimes (determined by how the communication and computation times are scaled with the number of workers) and precisely characterized the total run-time of the distributed algorithm and the optimum coding strategy in each regime. We further showed that coded schemes are $\Theta(\log(n))$ times faster than uncoded schemes in the regime of practical interest.

The interference model in this paper is not the most general model in wireless networks. It would be interesting to consider the case where the dependency of workers is characterized by their subsets that can be activated simultaneously. In this case, since the symmetry of workers will be broken, clearly one needs to assign different workload to different workers, which makes finding the optimal coding strategy and characterizing the total run-time more challenging. Another interesting direction is to generalize the problem to heterogeneous setup where each worker node has different computation capability.

\bibliographystyle{ieeetr}
\bibliography{ref}

\end{document}